\documentclass[11pt]{article}

\usepackage[utf8]{inputenc}
\usepackage{amssymb,amsmath}
\usepackage{color}
\usepackage{graphicx,graphics}
\usepackage{mathtools}
\usepackage[english]{babel}

\usepackage{epsfig,url}
\usepackage{hyperref}
\usepackage{bbm,theorem}
\usepackage{a4wide}
\usepackage{enumerate}
\usepackage{relsize}
\usepackage{physics}
\usepackage[square,sort,comma,numbers]{natbib}

\usepackage{xcolor}

\usepackage{verbatim}

\providecommand{\keywords}[1]
{
\textbf{{Keywords:}} #1
}

\newtheorem{theorem}{Theorem}[section]
\newtheorem{definition}[theorem]{Definition}

\newtheorem{lemma}[theorem]{Lemma}
\newtheorem{proposition}[theorem]{Proposition}
\newtheorem{assumption}[theorem]{Assumption}

{\theorembodyfont{\upshape}
\newtheorem{remark}[theorem]{Remark}

}
\numberwithin{equation}{section}
\numberwithin{theorem}{section}

\newcommand{\qed}{\hfill$\Box$}
\newenvironment{proof}{\begin{trivlist}\item[]{\em Proof:}\/}{%
\qed\end{trivlist}}
\newenvironment{proofof}[1]{%
\begin{trivlist}\item[]{\em Proof of #1}\ }{\qed\end{trivlist}}

\newcommand{\Z}{{\mathbb Z}}

\newcommand{\R}{{\mathbb R}}
\newcommand{\C}{{\mathbb C\hspace{0.05 ex}}}
\newcommand{\N}{{\mathbb N}}
\newcommand{\T}{{\mathbb T}}

\newcommand{\cf}[1]{{\mathbbm 1}_{\{#1\}}}

\renewcommand{\epsilon}{\varepsilon}

\renewcommand{\abs}[1]{\left| #1 \right|}

\newcommand{\ci}{{\rm i}}

\newcommand{\rme}{{\rm e}}
\newcommand{\rmd}{{\rm d}}
\newcommand{\diff}{{\rm d}}

\renewcommand{\norm}[1]{\Vert #1\Vert}

	\newcommand{\regabs}[1]{\langle #1 \rangle}

\newcommand{\UHaddress}{\em University of Helsinki, Department of Mathematics and Statistics\\
	\em P.O. Box 68, FI-00014 Helsingin yliopisto, Finland}
\newcommand{\inner}[2]{\langle #1, #2 \rangle}
\newcommand{\email}[1]{E-mail: \tt #1}
\newcommand{\emailaleksis}{\email{aleksis.vuoksenmaa@helsinki.fi}}
\begin{document}

\title{Dynamics of the infinite discrete nonlinear Schrödinger equation}
\author{Aleksis Vuoksenmaa\thanks{\emailaleksis}\\[0.5em]
	\UHaddress
}
\date{\today}
\maketitle
\abstract{The discrete nonlinear Schrödinger equation on \(\Z^d\), \(d \geq 1\) is an example of a dispersive nonlinear wave system. Being a Hamiltonian system that conserves also the \(\ell^2(\Z^d)\)-norm, the well-posedness of the corresponding Cauchy problem follows for square-summable initial data. In this paper, we prove that the well-posedness continues to hold for much less regular initial data, namely anything that has at most a certain power law growth far away from the origin. The growth condition is loose enough to guarantee that, at least in dimension \(d=1\), initial data sampled from any reasonable equilibrium distribution of the defocusing DNLS satisfies it almost surely.}

\vspace{5mm}


%\tableofcontents

\section{Introduction}
The discrete nonlinear Schrödinger equation (DNLS) is a nonlinear differential equation that describes the evolution of a complex valued field\footnote{By field, we mean a function \(\psi\colon X \to Y\) from a space of sites \(X\) to a space of values (here \(Y = \C\)). This can be also conceptualized as a sequence indexed by \(X\), i.e. \(\psi \in Y^{X}\).} \(\psi \colon \R_+ \times \Z^d \to \C\) in a \(d\)-dimensional space. For all times \(t \in \R_+\), each lattice site \(x \in \Z^d\) is associated with a value \(\psi_t(x)\), and the time evolution of \(t \mapsto \psi_t(x)\) is given by
\begin{align}
\ci \dv{t}\psi_t(x) = \sum_{y\in \Z^d}\alpha(x-y)\psi_t(y) + \lambda \abs{\psi_t(x)}^2 \psi_t(x), \quad x\in \Z^d.
\label{eq:NLS}
\end{align}

It is instructive to compare this with the more well-known continuous (cubic) nonlinear Schrödinger equation (NLS), a partial differential equation where the time evolution is given by
\begin{align}
\ci \partial_t \psi(t,x) &= -\Delta \psi(t,x) + \lambda \abs{\psi(t,x)}^2\psi(t,x),
\label{eq:continuous-NLS}
\end{align}
with \(\psi_t \colon \T^d \to \C\) or \(\psi_t \colon \R^d \to \C\). This equation has become an important example of dispersive equation, and it arises in a variety of different contexts, including as a description of the evolution of a condensate \cite{erdos_derivation_2007,kirkpatrick-et-al_2011_derivation}.

One way to see the discrete nonlinear Schrödinger equation is as a finite grid approximation of the continuum nonlinear Schrödinger equation \eqref{eq:continuous-NLS}. As opposed to the continuous model, the discrete nonlinear Schrödinger equation \eqref{eq:NLS} is an infinite system of differential equations, where the evolving variables are coupled to each others by a function \(\alpha \colon \Z^d \to \R_+\), which is called the hopping potential. The hopping potential plays the role of the Laplace operator in the discrete setting. In particular, if we set 
\begin{align}\alpha(y) = \left(\cf{\abs{y}_\infty = 0} - \frac{1}{2d}\cf{\abs{y}_\infty = 1}\right)
\label{eqdef:standard-discrete-laplacian}
\end{align} which is the standard discrete Laplacian. In this article, we allow for more general hopping potentials. We assume the hopping potential \(\alpha\) to have a finite support and to be symmetric. 

While the hopping potential couples the time-evolution of the field at a given site to the value of the field at nearby sites, \eqref{eq:NLS} also contains an onsite nonlinearity, whose strength is determined by the parameter \(\lambda \in \R\). If \(\lambda < 0\), the equations \eqref{eq:NLS} with \(\alpha\) given by \eqref{eqdef:standard-discrete-laplacian} and its continuum counterpart \eqref{eq:continuous-NLS}  are known as the \emph{focusing} DNLS and NLS, respectively. If \(\lambda > 0\), these are known as the \emph{defocusing} DNLS and NLS. As was already noticed in the pioneering work of Lebowitz, Rose, and Speer \cite{lebowitz_statistical_1988}, the behaviour of the equilibrium states of the NLS (both on \(\T^d\) and \(\R^d\)) depends drastically on the sign of \(\lambda\) -- see \cite{brydges_statistical_1996, bourgain_invariant_1996, bourgain_invariant_2000} for some results along these lines. On the discrete side, \cite{chatterjee2012probabilistic} studied the equilibrium states of the focusing DNLS.

For the deterministic results in this article, the sign and magnitude of the nonlinearity factor \(\lambda \neq 0\) will not play any role. Both the sign and the magnitude will, however, have a big effect on other aspects of the solutions. Since equilibrium measures are more readily defined for the defocusing DNLS, in the probabilistic statements in Section \ref{sec:random-initial-data} we will assume \(\lambda\) to be strictly positive.

Equation \eqref{eq:NLS} also acts as a possible starting point for the derivation of Boltzmann-like equations for wave systems. If the deterministic evolution prescribed by \eqref{eq:NLS} is coupled with random initial data, then the time-evolved field \(\psi_t\) becomes a random field. It is conjectured that the correlations of this random field satisfy a Boltzmann-like equation in the limit \(\lambda \to 0\), \(\tau \sim \lambda^{-2}t\), where \(\tau\) is the mesoscopic time scale, obtained by speeding up the microscopic dynamics by a factor of \(\lambda^{-2}\). Partial progress towards this result was achieved by Lukkarinen and Spohn in \cite{lukkarinen2011weakly}, where the initial data is assumed to follow and equilibrium distribution. In that article, rather than studying directly the correlation structure of \eqref{eq:NLS}, the problem is first turned into a finite version:
\begin{align}
	\ci \dv{t}\psi_t(x) = \sum_{y\in \Lambda_L}\alpha_{L}(x-y)\psi_t(y) + \lambda \abs{\psi_t(x)}^2\psi_t(x), \quad x\in \Lambda_L
	\label{eq:NLS-finite},
\end{align}
where \(\Lambda_L\) is a finite periodic box. After this, the behaviour of the \(L\)-dependent time-correlation structure can be studied for large \(L\) or in the limit \(L \to \infty\) without appealing to the full dynamics given by \eqref{eq:NLS}. In any case, \eqref{eq:NLS} and its finite version \eqref{eq:NLS-finite} acts as fruitful models for studying the derivation of wave kinetic equations, and understanding the well-posedness of \eqref{eq:NLS} and its connection to the approximation \eqref{eq:NLS-finite} will be of use in trying to connect the approach of \cite{lukkarinen2011weakly} to the more direct approach, where the initial values are sampled from an infinite dimensional distribution.

The well-posedness for \eqref{eq:NLS} is known for \(\ell^2(\Z^d)\) initial data, see \cite{hong2019strong}, which covers also fractional versions of the discrete Laplacian. In dimension \(d=1\) and with the standard Laplacean, the existence and uniqueness of solutions with \(\ell^\infty(\Z)\) initial data, together with estimates on the growth of the \(\ell^\infty\)-norm, were proven in \cite{dodson_nonlinear_2020}. Limiting to \(\ell^\infty\) initial data are still somewhat restrictive, and in the aforementioned applications (kinetic theory or the study of the equilibrium states), the initial data \(\psi_0\) for \eqref{eq:NLS} is sampled from a measure that makes the probability of \(\psi_0 \in \ell^\infty(\Z^d)\) (and hence that of \(\psi_0 \in \ell^2(\Z^d)\)) zero. An illustrative example is to consider a Gaussian process indexed by \(\Z^d\) with a translation invariant correlation structure, or any reasonable equilibrium measure of the system.

For this reason, the initial data for which the existence and uniqueness of the dynamics given by \eqref{eq:NLS} holds needs additional work. If the initial data grows fast enough as \(\abs{x}\to\infty\), then at least uniqueness might be lost. In between these too quickly growing classes of initial data and the \(\ell^\infty(\Z^d)\) initial data, there are hopefully classes of initial data where the problem has a unique solution, and the solution of \eqref{eq:NLS} is well-approximated by the finite box solutions \eqref{eq:NLS-finite} as \(L \to\infty\). Showing this is the main objective of the paper.

Techniques for studying the dynamics of such Hamiltonian systems on a lattice go back to the work of Lanford et al. \cite{lanford1977time}, where results concerning the dynamics of a family of anharmonic systems are proved. These tools may be adapted to our work with some modifications -- in particular the proofs of existence and uniqueness of the dynamics, as well as showing that sets of good initial data have full probability with respect to certain measures -- relies heavily on their proof strategy. On the other hand, in \cite{butta2007propagation, butta2016dynamics} the propagation of perturbations in anharmonic systems is studied, and it is here that we find techniques that can be adapted in order to show that the infinite box dynamics \eqref{eq:NLS} is well approximated by \eqref{eq:NLS-finite}.

As we already noted, we are interested in initial data that is typically not in \(\ell^2(\Z^d)\). The energy of the initial data will also typically be infinite. This has the upshot that we cannot use the standard Hamitonian system techniques to show that solutions exist and are uniquely defined. Instead, we have to formulate a local notion of the conservation laws, which can be shown to hold. The novel insight in the present paper is to use cancellations arising from symmetries of the free part of the Hamiltonian to control the time evolution of \emph{local} \(\ell^2\)-densities. This provides a way of showing that suitable weighted supremum norms propagate in the time-evolution.

Finally, it should be noted that there is an interesting -- although still not completely understood -- connection between between the DNLS and the NLS. Starting from \eqref{eq:NLS} (DNLS), \eqref{eq:continuous-NLS} (NLS) is formally attainable by taking a continuum limit, where the lattice spacing goes to zero and the graph laplacian is blown up to yield the standard laplacian in the limit. Proving this remains a formidable problem.  This was done for the one dimensional DNLS as well as for long range hopping potentials in \cite{kirkpatrick2013continuum}, and results along the same line with stronger notion of convergence were obtained in \cite{hong2019strong}. Such results are very interesting, but in this paper nothing to that direction is proven, as we hold the lattice spacing constant (corresponding to UV cutoff).

Section \ref{sec:setup-and-problem-statement} starts with introducing the basic setup. We will discuss the relation between the finite periodic system \eqref{eq:NLS-finite} and the full dynamics \eqref{eq:NLS} and describe how we can treat the second as a limit of the first one. There we will also state the main results, namely Theorem \ref{thm:existence-uniqueness-approximation} and Theorem \ref{thm:almost-sure-initial-data}. In section \ref{sec:control-of-local-particle-number}, we prove that certain local \(\ell^2\)-norms propagate in time. This result is the applied in sections \ref{sec:existence-of-dynamics}, \ref{sec:uniqueness-of-dynamics} and \ref{sec:approximation-of-dynamics}, where we will prove the main lemmas: existence, uniqueness, and approximability of solutions to \eqref{eq:NLS}. Together, the main lemmas imply Theorem \ref{thm:existence-uniqueness-approximation}. Finally, in Section \ref{sec:random-initial-data}, we will discuss different probability measures for the initial data and prove Theorem \ref{thm:almost-sure-initial-data}.

\section*{Acknowledgements}

I want to thank Jani Lukkarinen for an introduction to the topic and for
several useful discussions, as well as for his comments on the text. I
am also grateful to Herbert Spohn for providing his preliminary notes about how the results of Buttà, Caglioti, de Ruzza and Marchioro might be applicable in the present case.  Finally, I would like to thank Kalle
Koskinen for comments related to the last section of the article.

This research has been supported by the Academy of Finland, via an Academy project (project No. 339228) and the Finnish {\em centre of excellence in Randomness and Structures\/} (project No. 346306).

\section{Setup and main result}
\label{sec:setup-and-problem-statement}

It is evident from \eqref{eq:NLS} that the full DNLS on \(\Z^d\) is a system of infinitely many differential equations, where each variable indexed by \(\Z^d\) is coupled to nearby variables through the hopping potential \(\alpha\). In addition, each variable has a nonlinear onsite term, which in this case is cubic. In order to study \eqref{eq:NLS} and the corresponding Cauchy problem rigorously, we begin by making the system finite by considering a box with a finite side-length. In other words, we move to study the system \eqref{eq:NLS-finite}. Here \(\Lambda_L\) is a finite box in \(\Z^d\), and the addition of points is taken modulo the side length of the box. We will describe this more precisely below, but in essence it turns the full system \eqref{eq:NLS} into a finite system \eqref{eq:NLS-finite} with periodic boundary conditions. The problem of showing that the dynamics \eqref{eq:NLS} exist is then reduced proving that these finite dynamics have a suitable limit as the side length \(L\) goes to infinity.

To this end, we define the periodic lattice of side-length \(2L+1\) as
\begin{align}
	\Lambda_L \coloneqq \{-L,\dots, 0, \dots, L\}^d.
\end{align}
Restricting to lattices of odd side-length simplifies the notation, but does not change the remaining argument. The moniker \emph{periodic} comes from the fact that the addition of two points on the lattice is defined periodically:
\begin{align}
	x +_{\Lambda_L} y = x + y \mod 2L+1, \quad x, y \in \Lambda_L \subset \Z^d  
\end{align}
Moving to the finite periodic problem allows us to define a finite dual lattice, and to diagonalize the linear part of the dynamics using the Fourier transform. However, this technique will not be used here.

We therefore turn our attention to the following initial value problem
\begin{align}
	\begin{cases}
		\ci \dv{t}\psi_t(x) &= \sum_{y\in \Lambda_L} \alpha_L(x-y)\psi_t(y) + \lambda \abs{\psi_t(x)}^2\psi_t(x), \quad x \in \Lambda_L \\
		\psi_t(x)\vert_{t=0} &= \psi_0(x), \quad x \in \Lambda_L.
	\end{cases}
	\label{problem:DNLS-finite-cauchy}
\end{align}

This is a finite dimensional system of partial differential equations, and stands in contrast with the infinite dimensional initial value problem, corresponding to the full dynamics:
\begin{align}
	\begin{cases}
		\ci \dv{t}\psi_t(x) &= \sum_{y\in \Z^d} \alpha(x-y)\psi_t(y) + \lambda \abs{\psi_t(x)}^2\psi_t(x), \quad x \in \Z^d \\
		\psi_t(x)\vert_{t=0} &= \psi_0(x), \quad x \in \Z^d.
	\end{cases}
	\label{problem:DNLS-cauchy}	
\end{align}

\begin{definition}
	\label{def:solution-to-DNLS}
	We say that a time-dependent sequence \(\psi \colon \R_+ \to (\Z^d \to \C)\) is a solution to the initial value problem \eqref{problem:DNLS-cauchy}, in case for each \(x\in \Z^d\), the mapping \(t \mapsto \psi_t(x)\) is continuously differentiable and satisfies \eqref{problem:DNLS-cauchy}.
\end{definition}

Here the hopping potential \(\alpha_L\colon\Lambda_L\to \R\) is defined in terms of the hopping potential \(\alpha \colon \Z^d \to \R\) by precomposing \(\alpha\) with the embedding \(\Lambda_L \hookrightarrow \Z^d\). We will typically drop the subscript from the notation. At this point, we make the following assumption on the hopping potential.

\begin{assumption}
	The hopping potential \(\alpha\) is of finite range, i.e. there exists \(\ell \geq 0\) such that \(\alpha(x) = 0\) for all \(x \in \Z^d\) with \(\abs{x}_\infty \geq \ell\).
\end{assumption}
One class of such hopping potentials are the nearest neighborhood potentials, where the interaction range is bounded by \(1\). This class covers the standard discrete Laplace operator, but also many others.

The above definitions give full meaning to the finite DNLS, as described by \eqref{eq:NLS-finite}. It is easy to see that \eqref{eq:NLS-finite} can be realized as a Hamiltonian system, with the energy function given by
\begin{align}
	H_L(\psi) = \sum_{x,y\in \Lambda_L}\alpha(x-y)\psi(x)\psi(y)^* + \frac{\lambda}{2}\sum_{x\in \Lambda_L}\abs{\psi(x)}^4.
	\label{eq:finite-hamiltonian}
\end{align}
Here \(\psi(x)^*\) is the complex conjugate of \(\psi(x)\). The canonical pair of variables is given by \(\Re \psi(x)\) and \(\Im \psi(x)\).

\begin{comment}Alternatively, we can define the dynamics on a two-component field by letting
\[H(\psi) = \sum_{x,y\in \Lambda}\alpha_L(x-y)\psi(x,+)\psi(y,-) + \frac{\lambda}{2}\psi(x,+)^2\psi(x,-)^2\] for \(\psi \colon \Lambda \times \{\pm\} \to\C\) and by studying the evolution given by \(\ci \dv{t}\psi_t(x,\sigma) = \sigma \partial_{(x,-\sigma)}H(\psi_t	)\), where \(\partial_{(x,\sigma)}\) denotes taking the complex partial derivative with respect to the \((x,\sigma)\)-component of \(\C^{\Lambda \times \{\pm\}}\). If we require that the initial data satisfies \(\psi_0(x,-) = \psi_0(x,+)^*\), then it is easy to see that this condition also holds for later times and matches the dynamics of \eqref{eq:NLS-finite}. We will move between these two notation, depending on what is more convenient in any given context.
\end{comment}

A direct calculation shows that both the \(\ell^2(\Lambda_L)\)-norm \(\norm{\cdot}_{\ell^2(\Lambda_L)}^2\) and the energy \(H_L\) are conserved by the time-evolution: \(\dv{t} \norm{\psi_t}^2_{\ell^2(\Lambda_L)} = 0\) and \(\dv{t} H_L(\psi_t) = 0.\) With this property it is evident that \eqref{problem:DNLS-finite-cauchy} admits a unique, continuosly differentiable solution to any initial value \(\psi_0\).

Each finite sequence \(\psi^L \colon \Lambda_L \to \C\) defined on the periodic lattice can also be realized as an infinite infinite sequence \(\Z^d \to \C\) by lifting the \(L\)-periodicity into an \(L\)-dependent translation invariance condition. Given a periodic sequence \(\psi^L \colon \Lambda_L \to \C\), let define a new sequence \(\psi \colon \Z^d \to \C\) by letting
\[\psi(x + (n_1,\dots,n_d)L) \coloneqq \psi^L(x), \quad x \in \Lambda_L\] for all \((n_1,\dots, n_d) \in \Z^d\). For a fixed \(L\), this mapping from \(\Lambda_L \to\C\) sequences to \(\Z^d \to \C\) sequences is injective. On the other hand, each function \(\psi \colon \Z^d \to \C\) can be truncated into a finite sequence \(\psi^L \colon \Lambda_L \to \C\) for any \(L \geq 1\). Finally, we can recover \(\psi\) pointwise as a limit of these truncations as \(L\to\infty\). Therefore, we have the following operations
\begin{align}
	\texttt{Trunc}_L &\colon (\Z^d \to \C) \to (\Lambda_L \to \C) \\
	\texttt{Embed}_L &\colon (\Lambda_L \to \C) \to (\Z^d \to \C)
\end{align}
and these operations are realize the periodic boundary conditions in the sense that the function
\[(\texttt{Embed}_L \circ \texttt{Trunc}_L)(\psi) \colon \Z^d \to \C\] is always \(L\)-periodic, and
\[(\texttt{Trunc}_L \circ \texttt{Embed}_L)(\psi^L) = \psi^L\] for any finite sequence \(\psi^L \colon \Lambda_L \to \C\). Given \(\psi \colon \Z^d\to\C\), we typically use the notation \(\psi^L \coloneqq \texttt{Trunc}_L(\psi)\), and it should be clear from the context what we mean by this.

Looking at the finite DNLS \eqref{eq:NLS-finite}, it is clear that with \(\alpha\) of finite range, we can write the evolution of each solution \(\psi_t(x)\) in terms of a local polynomial of the field, i.e.
\begin{align}
	\dv{t}\psi_t(x) = -\ci G_x(\psi_t), \quad x\in \Lambda_L,
\end{align}
where \(G_x\) sees the value of the field \(\psi_t\) only at sites that are within the range of \(\alpha\) from \(x\), taking into account the periodic structure of \(\Lambda_L\).

From \eqref{eq:NLS-finite}, it follows also that the complex conjugate evolves as
\begin{align}
	\dv{t}\psi_t(x)^* &= \ci \sum_{y\in \Lambda_L}\alpha(x-y)\psi_t(y)^* + \ci \lambda \abs{\psi_t(x)}^2\psi_t(x)^* \\
	&= \ci G_x(\psi_t)^*, \quad x \in \Lambda_L.
\end{align}
assuming that \(\psi_t\) solves pointwise \eqref{eq:NLS-finite}, we would then obtain
\begin{align}
	\dv[2]{t}\psi_t(x) = &- \sum_{y\in \Lambda_L}\sum_{z\in \Lambda_L}\alpha(x-y)\alpha(y-z)\psi_t(z) \nonumber \\
	&-\lambda \sum_{y\in \Lambda_L}\alpha(x-y)\abs{\psi_t(y)}^2\psi_t(y) -2\lambda \sum_{y\in \Lambda_L}\alpha(x-y)\psi_t(y)\abs{\psi_t(x)}^2 \nonumber \\  &-\lambda^2 \abs{\psi_t(x)}^4\psi_t(x) \nonumber +\lambda \psi_t(x)^2\sum_{y\in \Lambda_L}\alpha(x-y)\psi_t(y)^*	
\end{align}

So in short, we can write
\begin{align}
	\dv[2]{t}\psi_t(x) = P^x(\psi_t),
\end{align}
where \(P^x\) is a local polynomial of order \(\leq 5\) and it depends only on \(\psi_t(z)\) that satisfy \[\abs{x-z}_\infty \leq 2\ell.\] 

If \(\psi \colon \R_+ \to \C^{\Z^d}\) is a solution of \eqref{eq:NLS} in the sense of Definition \ref{def:solution-to-DNLS}, then it solves the following Duhamel iterations of first and secod order:
\begin{align}
	\psi_t(x) = \psi_0(x) + \int_{0}^{t}(-\ci G^x(\psi_s)) \rmd s, \quad t\in \R_+, \quad x \in \Z^d
	\label{eq:integral-form-1-order},
\end{align}
and
\begin{align}
\psi_t(x) = \psi_0(x) - \ci tG^x(\psi_0) + \int_{0}^{t}(t-s)P^x(\psi_s)\rmd s, \quad t \in \R_+, \quad x \in \Z^d
\label{eq:integral-form-2-order}.
\end{align}
Formulation \eqref{eq:integral-form-1-order} is equivalent with solving \eqref{eq:NLS} in the sense that if \(\psi \colon \R_+ \to\C^{\Z^d}\) is such that \(t\mapsto \psi_t(x)\) is continuous for every \(x \in \Z^d\) and solves \eqref{eq:integral-form-1-order}, then \(t \mapsto \psi_t(x)\) is also differentiable for every \(x \in \Z^d\) and \(\psi\) satisfies \eqref{eq:NLS} with the initial condition \(\psi_0\). This reduces the study of the infinite system of differential equations into solving the fixed point problem \eqref{eq:integral-form-1-order}. We will use this in proving the existence of solutions to \eqref{eq:NLS}. To prove that we can find a solution to \eqref{eq:integral-form-1-order}, we start by looking at the finite, \(L\)-periodic problem. From these solutions, we then obtain a subsequence of side lenghts \(L\) and a limit point \(\psi\), such that \(\psi_t(x) = \lim_{k\to\infty}\psi^{L_k}_t(x)\) and \(\psi_t^{L_k}\) are global solutions to \ref{problem:DNLS-finite-cauchy} with side length \(L_{k}\). After this, we simply note that the limit point \(\psi\) satisfies \eqref{eq:integral-form-1-order}.

The second order Duhamel expansion \eqref{eq:integral-form-2-order} comes handy in proving the uniqueness of the solutions. Since any solution to \eqref{eq:NLS} in the sense of Definition \eqref{def:solution-to-DNLS} must satisfy \eqref{eq:integral-form-2-order}, we can use the latter formula to show that the difference between two candidate solutions with same initial data is small.

In order to prove that solutions to \eqref{eq:NLS} exist and that they are unique, we will need to study the problem in suitable sequence spaces. To this end, we use the following definition and notation.
\begin{definition}[Weighted sequence spaces]
	Let \(\Phi \colon \Z^d \to (0,+\infty)\) be a weight function. For a given sequence \(\psi \colon \Z^d \to \C\), we say that \(\psi \in X_{\Phi}\), in case it satisfies \(\sup_{x\in \Z^d} \abs{\psi(x)}\Phi(x) < +\infty\).
\end{definition}
These spaces are all Banach spaces. Two classes of such weighted sequence spaces are of special interest: exponentially bounded sequences, and sequences bounded by a power law.

\begin{definition}[Exponentially bounded sequences]
	\label{def:exponentially-bounded-seq}
	
	For \(q > 0\), we say that the space of \(q\)-exponentially bounded sequences is given by those \(\psi \colon \Z^d \to \C\) for which 
	\begin{align}
		\sup_{x\in \Z^d} \rme^{-q\abs{x}_\infty}\abs{\psi(x)} < +\infty.
	\end{align}
	Such space is denoted by \(U_q = U_q(\Z^d)\), and its norm by \(\norm{\cdot}_{U_q}\).
\end{definition}

In the next definition, we use the standard notation \(\regabs{x} = (1+\abs{x}^2)^{1/2}\).
\begin{definition}
	\label{def:power-law-bounded-seq}
	For \(p > 0\), the space of \(p\)-power law bounded sequences is given by
	\begin{align}
		\sup_{x\in \Z^d} \regabs{x}^{-p}\abs{\psi(x)} < +\infty.
	\end{align}
Such space is denoted by \(X_p\) and its corresponding norm by \(\norm{\cdot}_p\).
\end{definition}
Here we note that \(X_p \subset X_{p'}\) for \(p < p'\), \(U_{q} \subset U_{q'}\) for \(q<q'\), and \(X_p \subset U_q\) for all \(p, q\).

With these definitions, we are finally able to state the main result of the article.
\begin{theorem}
	\label{thm:existence-uniqueness-approximation}
	Let \(c>0\) be arbitrarily small. For \(\psi_0 \in X_{1/2-c}\), there exists a one-parameter group of transformations \((\phi_t)_{t\geq 0}\), such that \(t \mapsto \psi_t = \phi_t(\psi_0)\) is the unique global solution to \eqref{eq:NLS}. This can be approximated in terms of the sequence of solutions of \ref{eq:NLS-finite}, when the initial data is taken to be \(\psi^L_0 = \textrm{Trunc}_{L}(\psi_0)\).
\end{theorem}
\begin{proof}
	This follows from Lemmas \ref{lemma:existence-of-solutions}, \ref{lemma:uniqueness-of-solutions} and \ref{lemma:approximation}, which will be proven below.
\end{proof}

\section{Control of the local particle number density}
\label{sec:control-of-local-particle-number}
To obtain useful bounds for the evolution of some weighted supremum norms of the field, we move from the conserved quantities \(\norm{\cdot}_{2}\) and \(H\) (which may be infinite in the infinite volume limit) to quantities that are not completely conserved, but are always finite and whose growth rate is bounded from above. To this end, define the local particle number, centered around \(x\in \Lambda_L\) and with weight \(\epsilon\), as
\[N_{\epsilon, x}(\psi_t) \coloneqq \sum_{y \in \Lambda_L} \rme^{-\epsilon \abs{x-y}_\infty}\abs{\psi_t(y)}^2 = \sum_{y \in \Lambda_L} \rme^{-\epsilon \abs{y}_\infty}\abs{\psi_t(x+y)}^2\]
and the local particle number density as
\[Q_{\epsilon, x}(\psi_t) \coloneqq \frac{1}{S_\epsilon} N_{\epsilon, x}(\psi_t), \quad S_{\epsilon} = \sum_{x \in \Lambda_L} \rme^{-\epsilon \abs{x}_\infty}.\]

A direct calculation shows that
\begin{align}
\partial_t Q_{\epsilon, x}(\psi_t) 
&= \frac{1}{S_{\epsilon}} \sum_{y\in \Lambda_L} \rme^{-\epsilon\abs{x-y}_\infty} F_{y,N}(\psi_t) \\
&=\frac{1}{S_\epsilon}M_{\epsilon}(x, t)
\end{align}
where
\begin{align}
F_{y,N}(\psi_t) &= \sum_{x'\in \Lambda_L}\alpha(y-x')\ci (\psi_t(y)\psi_t(x')^* - \psi_t(y)^* \psi_t(x'))\\
M_\epsilon(x,t) &=  \sum_{y\in \Lambda_L}\rme^{-\epsilon \abs{x-y}_\infty} F_{y,N}(\psi_t)
\end{align}

If we can control the density of \(M_{\epsilon}(\cdot, t)\), then we have control over the evolution of the local particle number density. It should be noted that \(Q_{\epsilon}(x,t)\) and \(M_{\epsilon}(x,t)\) are both real quantities, and \(Q_{\epsilon}(x,t)\) is positive.

Due to the antisymmetric role that the variables \(y\) and \(x'\) take in the formula of \(M_{\epsilon}(x,t)\), we can also write it as
\begin{align}
	M_{\epsilon}(x, t) = \frac{1}{2} \sum_{y, x' \in \Lambda} \left(\rme^{-\epsilon \abs{x-y}_\infty}- \rme^{-\epsilon \abs{x-x'}_\infty}\right) \alpha(y-x') \ci (\psi_t(y)\psi_t(x')^* - \psi_t(y)^* \psi_t(x')).
\end{align}

The benefit of the above formula is that it makes clear that there are large cancellations in the local particle number densities. Taking \(\epsilon = 0\), we see that \(M_{0}(x, t) = 0\) for all \(x \in \Lambda\) and \(t \geq 0\). For positive \(\epsilon\) we migth have \(M_{\epsilon}(x,t)\neq 0\). However, since \(x' \mapsto \alpha(y-x')\) is concentrated near \(y \in \Lambda\), the exponential terms \(\rme^{-\epsilon \abs{x-y}_\infty}\) and \(\rme^{-\epsilon \abs{x-x'}_\infty}\) are comparable in size. Therefore, large cancellations should make \(M_{\epsilon}(x,t)\) small enough. This is the content of the next lemma.
\begin{lemma}
	For any \(x \in \Lambda\), \(\epsilon \in (0, \frac{1}{2\ell})\) and \(t \geq 0\), we have
	\begin{align}\label{ineq:time-evolution-of-local-particle-number-density}
	\abs{Q_{\epsilon,x}(\psi_t)} \leq \rme^{\tilde{\epsilon} t}\abs{Q_{\epsilon,x}(\psi_0)},
	\end{align}
	where \(\tilde{\epsilon}\) depends on the localization scale \(\epsilon\), the dimension \(d\) of the system, and the hopping potential \(\alpha\).
\label{lemma:control-of-local-particle-number}
\end{lemma}
\begin{proof}
		We partition \(\Lambda^2_L = \Lambda_L \times \Lambda_L\) into two sets:
		\begin{align*}
		&\Gamma^L_1(x) \coloneqq \{(y, x') \in \Lambda^2_L \colon \abs{y-x}_\infty \geq \abs{x'-x}_\infty\} \\
			&\Gamma^L_2(x) \coloneqq \{(y,x') \in \Lambda^2_L \colon \abs{y-x}_\infty < \abs{x'-x}_\infty\}.
		\end{align*}
		
		The contribution to \(\abs{M_\epsilon(x,t)}\) coming from the sum over \(\Gamma_1^L(x)\) is bounded from above by
		\begin{align*}
			\frac{1}{2}&\abs{\sum_{(y, x') \in \Gamma^L_1(x)} \left(\rme^{-\epsilon \abs{x-y}_\infty} - \rme^{-\epsilon \abs{x-x'}_\infty}\right) \alpha(y-x') \ci (\psi_t(y)\psi_t(x')^* - \psi_t(y)^* \psi_t(x'))} \\
			&\leq c \left(\sum_{(y,x') \in \Gamma^L_1(x)} \rme^{-\epsilon \abs{x-x'}_\infty}\left(1- \rme^{-\epsilon\left(\abs{x-y}_\infty - \abs{x-x'}_\infty\right)}\right) \abs{\alpha(y-x')}\abs{\psi_t(y)}^2\right)^{1/2} \\
			&\times  \left(\sum_{(y,x') \in \Gamma^L_1(x)} \rme^{-\epsilon \abs{x-x'}_\infty}\left(1- \rme^{-\epsilon\left(\abs{x-y}_\infty - \abs{x-x'}_\infty\right)}\right) \abs{\alpha(y-x')}\abs{\psi_t(x')}^2\right)^{1/2}
			\end{align*}
	If \(x'\) and \(y\) are from a distance \(\abs{x'-y}_\infty > \ell\) from each others, we have \(\alpha(x'-y) = 0\), so such pairs will not contribute to the sum. The right hand side of this can thus be bounded from above by
	\begin{align*}
	&c \epsilon^{1/2} \ell^{1/2} \left(\sum_{(y,x') \in \Gamma^L_1(x)} \rme^{-\epsilon \abs{x-x'}_\infty} \abs{\alpha(y-x')} \abs{\psi_t(y)}^{2}\right)^{1/2} \\
			&\times \epsilon^{1/2}\ell^{1/2} \left(\sum_{(y,x') \in \Gamma^L_1(x)} \rme^{-\epsilon \abs{x-x'}_\infty} \abs{\alpha(y-x')} \abs{\psi_t(x')}^{2}\right)^{1/2} \\
			&\leq c \epsilon \ell \left(\sum_{y \in \Gamma_1^L(x)} \rme^{\epsilon \ell}\rme^{-\epsilon \abs{x-y}_\infty} \abs{B^d(0,\ell)} \norm{\alpha}_\infty \abs{\psi_t(y)}^{2}\right)^{1/2} \\
			&\times \left(\sum_{x' \in\Gamma^L_1(x)} \rme^{-\epsilon \abs{x-x'}_\infty} \abs{B^d(0, \ell)} \norm{\alpha}_\infty \abs{\psi_t(x')}^{2}\right)^{1/2} \\ 
			&\leq C_{d, \alpha, \ell} \epsilon N_{\epsilon,x}(\psi_t)^{1/2} N_{\epsilon, x}(\psi_t)^{1/2} = \epsilon C_{d, \alpha, \ell}  N_{\epsilon,x}(\psi_t),
		\end{align*}
		where we can pick the explicit bound \(C_{d,\alpha,\ell} = c \ell \rme^{\epsilon \ell/2} \norm{\alpha}_\infty (2\ell +1)^d\)

	In combination with a symmetric calculation in the case that \(\abs{y-x}_\infty < \abs{x'-x}_\infty\), the above calculation shows that
	\begin{align*}
	&\abs{\sum_{(y, x') \in \Gamma_2(x)} \left(\rme^{-\epsilon \abs{x-y}_\infty} - \rme^{-\epsilon \abs{x-x'}_\infty}\right) \alpha(y-x') \ci (\psi_t(y)\psi_t(x')^* - \psi_t(y)^* \psi_t(x))} \\
	&\leq \epsilon C_{d, \alpha, \ell} N_{\epsilon, x}(\psi_t).
	\end{align*}
	We divide both sides by \(S_{\epsilon}\) and use the positivity of \(Q_{\epsilon,x}(\psi_t)\) to get the inequality
	\begin{align*}
	\abs{\partial_t Q_{\epsilon, x}(\psi_t)} \leq \epsilon C_{d, \alpha, \ell}Q_{\epsilon, x}(\psi_t).
	\end{align*}
	This shows that
	\begin{align}
	\abs{Q_{\epsilon,x}(\psi_t)} \leq \int_{0}^t \epsilon C_{d,\alpha, \ell} \abs{Q_{\epsilon,x}(\psi_s)} \diff s + \abs{Q_{\epsilon, x}(\psi_0)},
	\end{align}
	so by Grönwall and the positivity of \(Q_{\epsilon,x}(\psi_t)\), it follows that
	\begin{align}
	Q_{\epsilon,x}(\psi_t) \leq \rme^{\tilde{\epsilon} t} Q_{\epsilon,x}(\psi_0),
	\end{align}
	where \(\tilde{\epsilon} = \epsilon \cdot C_{d,\alpha, \ell}\).
	\end{proof}
	\begin{remark}
		One way to think of this inequality is that it shows the the local particle number of the system cannot grow too fast. In fact, it can grow at most exponentially in time. This rate is determined by by localization scale \(\epsilon\), dimension of the system, and properties of the hopping potential.
	\end{remark}
	\begin{remark}
		If we take \(L\to+\infty\), then with suitably regular initial data both sides of \eqref{ineq:time-evolution-of-local-particle-number-density} are defined, so the inequality extends to the infinite lattice.
	\end{remark}

From \ref{lemma:control-of-local-particle-number}, it follows immediately that the non-normalized local particle number \(N_{\epsilon, x}(\psi_t)\) satisfies 
\begin{align}
N_{\epsilon,x}(\psi_t) \leq \rme^{\tilde{\epsilon}t}N_{\epsilon,x}(\psi_0).
\end{align}
Here the left hand side controls \(\abs{\psi_t(x)}^2\) and the right hand side depends on the initial data only. Since \(\abs{\psi_t(x)}^2 \leq N_{\epsilon,x}(\psi_t)\), we also obtain
\begin{align}
	\abs{\psi_t(x)}^2 \leq \rme^{\tilde{\epsilon}t}N_{\epsilon,x}(\psi_0).
\end{align}

This implies the following bound for \(\abs{\psi_t(x)}\):
\begin{align}\label{ineq:pointwise-control-at-t}
	\abs{\psi_t(x)} \leq \rme^{\tilde{\epsilon} t}\sum_{y\in \Lambda_L}\rme^{-\frac{\epsilon}{2}\abs{x-y}_\infty}\abs{\psi_0(y)}.
\end{align}

Therefore, if \(\Phi \colon \Z^d \to \R_+\) is a weight function for which \(\psi_0\) is in the corresponding weighted space, we have
\begin{align}\label{ineq:weighted-pointwise-bound-at-t}
\abs{\psi_t(x)}\Phi(x) \leq \rme^{ \tilde{\epsilon}t} \sum_{y\in \Lambda_L} \rme^{-\frac{\epsilon}{2}\abs{x-y}_\infty}\frac{\Phi(x)}{\Phi(y)}\norm{\psi_0}_{\Phi}, \quad x \in \Lambda_L.
\end{align}
This means that the weighted supremum norm at time \(t\) is controlled in terms of the same norm of the initial data, and this bound becomes worse exponentially in time. Moreover, the prefactor of the bound is determined by
\begin{align}
\sup_{x \in \Lambda_L}\sum_{y\in \Lambda_L}\rme^{-\frac{\epsilon}{2}\abs{x-y}_\infty}\frac{\Phi(x)}{\Phi(y)},
\label{formula:prefactor-of-weighted-sup-norm}
\end{align}
where \(\epsilon < \frac{1}{2\ell}\), with \(\ell \geq 1\) being the coupling range of the hopping potential.

If we consider a sequence of finite initial value problems \eqref{problem:DNLS-finite-cauchy}, with the initial data given by the \(L\)-truncation of a given initial data \(\psi_0\), then uniform boundedness in \(L\) of \eqref{formula:prefactor-of-weighted-sup-norm} for a given \emph{decreasing weight} \(\Phi\) gives a bound for all times for the full lattice evolution, as the right hand side of the following inequality doesn't depend on \(L\):
\begin{align}
	\abs{\psi^L_t(x)}\Phi(x) \leq \rme^{\tilde{\epsilon} t}\left(\sum_{y\in \Z^d} \rme^{-\frac{\epsilon}{2}\abs{x-y}}\frac{\Phi(x)}{\Phi(y)}\right)\norm{\psi_0}_{\Phi}, \quad x \in \Z^d.
	\label{ineq:truncated-upper-bound}
\end{align}

\section{Existence of solutions}
\label{sec:existence-of-dynamics}

\begin{lemma}
	\label{lemma:existence-of-solutions}
	Let \(\psi_0 \in U_{q}\) for \(q \in (0, 1/(4\ell))\), where \(\ell\) is the range of the hopping potential. Then there exists a solution to the Cauchy problem \eqref{problem:DNLS-cauchy} in the sense of Definition \ref{def:solution-to-DNLS}.
\end{lemma}
\begin{proof}
We have fixed the initial value \(\psi_0 \in U_q\), so we can consider all the finite problems \eqref{problem:DNLS-finite-cauchy} with the truncated initial value \(\psi_0^L\). Note that if we embed \(\psi^L_0\) into \(\C^{\Z^d}\), we have \(\norm{\psi_0^L}_{U_q} \leq \norm{\psi_0}_{U_q}\). Fix \(T > 0\) to be a target time. The sequence of initial value problems parametrized by \(L \in \N\)
\begin{align}
	\begin{cases}
		&\ci \dv{t}\psi_t^L(x) = \sum_{y\in \Lambda_L} \alpha_L(x-y)\psi^L_t(y) + \lambda \abs{\psi_t^L(x)}^2\psi_t^L(x), \quad x \in \Lambda_L, t \in [0, T] \\
		&\psi_0^L(x) = \psi_0(x), \quad x \in \Lambda_L
	\end{cases}
\end{align}
can be formulated in the integral form as.
\begin{align}
	\psi^{L}_t(x) &= \psi^L_0(x) - \ci \int_{0}^t \rmd s G^x(\psi^L_s), \quad x \in \Lambda_L, t\in [0, T].
\end{align}

From this formulation, we find a continuously differentiable solution \[\psi^L \colon [0, T] \to (\Lambda_L \to \C).\] This solution is continuous from \([0, T]\) to \(\C^{\Lambda_L}\) with the supremum norm, so for each \(x \in \Lambda_L\), the function \(f^L_x \colon [0, T] \to \C\), \(f^L_x(t) = \psi^L_t(x)\) is continuous (and continuously differentiable). Furthermore, \(\abs{f^L_x(t)} = \abs{\psi^L_t(x)} \leq C_{T,x}\), where the constant is independent of \(L\). This follows from \eqref{ineq:truncated-upper-bound}. We also have equicontinuity in time with \(L\) being the varying parameter. By Arzela-Ascoli, we find a limit of a subsequence. Doing the diagonal trick on the collection \(x \in \Z^d\), we find a common subsequence so that the limits \(f_x \coloneqq \lim_{j\to\infty}f^{L_j}_x\) exist in \(C([0, T], \R)\) for every \(x\in \Z^d\), and this limiting configuration satisfies for all every \(x \in \Z^d\):
\begin{align}
	f_x(t) &= f_x(0) - \ci \int_{0}^t (\sum_{y \in \Z^d}\alpha(x-y)f_y(s) + \lambda \abs{f_x(s)}^2 f_x(s)) \rmd s, \quad t \in [0, T].
\end{align}

After yet another application of the diagonal trick, in this instance for the target times \(T \in \N\), we obtain a solution (denoted this time by \((\psi_t(x))_{x\in \Z^d}\) for each fixed time \(t\)) of \eqref{eq:integral-form-1-order}. This also solves the Cauchy problem
\begin{align}
	\begin{cases}
		\ci \dv{t}\psi_t(x) &= \sum_{y\in \Z^d} \alpha(x-y)\psi_t(y) + \lambda \abs{\psi_t(x)}^2 \psi_t(x), \quad x \in \Z^d, t \in \R_+ \\
	\psi_t(x)\vert_{t=0} &= \psi_0(x), \quad x \in \Z^d
	\end{cases}
\end{align}
in the sense of Definition \ref{def:solution-to-DNLS}.
\end{proof}
Furthermore, \(\psi_t\) inherits the weighted sup-norm estimates:
\begin{align}
	\sup_{x\in \Z^d} \rme^{-r \abs{x}}\abs{\psi_t(x)} \leq \rme^{ct}C_{d,r} \sup_{x\in \Z^d} \rme^{-r \abs{x}} \abs{\psi_0(x)}, \quad r \in (0, c/2).
\end{align}

\begin{remark}
Thus far, we have proven the existence of some limiting dynamics. Provided that we have uniqueness of the dynamics, it follows that
\begin{align}
\sup_{x\in \Z^d} \Phi(x)\abs{\psi_t(x)-\psi_s(x)} \leq \sup_{x\in \Z^d} \int_{s}^t \rmd r \Phi(x) \abs{G_x(\psi_r)}
\end{align}
Here \(G_x\) is a local polynomial of order \(3\). Therefore, if the norm satisfies uniformly a perturbation upper bound \(\Phi(x+y) \leq c \Phi(x)\) for \(\abs{y}_\infty \leq \tilde{c}\), and if \(G(\psi_r)\) is also bounded in the norm with the weight \(\Phi(x)^{1/3}\), then the integrand becomes bounded by \(\norm{\psi_r}_{\Phi^{1/3}}^3 \leq C_T\norm{\psi_0}_{\Phi^{1/3}}^3\), and the smallness coming from the size of the integration domain then guarantees that \(\psi \in C([0, T], X_\Phi)\) as well.
\end{remark}

\section{Uniqueness of solutions with power law bounds}
\label{sec:uniqueness-of-dynamics}

Lemma \ref{lemma:existence-of-solutions} shows that the dynamics exists for all initial data that is in some space of exponentially bounded sequences \(U_q\), with \(q <1/(4\ell)\), \(\ell \geq 1\) being the range of the hopping potential. The solution was constructed using compactness arguments, and it might not be unique. To prove the uniqueness of solutions, we have to require that the initial data satisfies additional regularity conditions, namely that the initial data is bounded by a power law with the exponent \(\frac{1}{2}\). In other words, we require that the initial data belongs to \(X_{1/2}\), i.e.
\begin{align}
	\sup_{x\in \Z^d} \regabs{x}^{-1/2}\abs{\psi(x)} < +\infty.
\end{align}
The space \(X_{1/2}\) is a Banach space, and we have \(X_{1/2} \subset U_q\) for every \(q > 0\). In particular, \(X_{1/2} \subsetneq \bigcap_{q>0}U_q\) as sets. This tells us that if we assume our initial data to belong to \(X_p\), there exists a solution to the Cauchy problem, as shown in \ref{lemma:existence-of-solutions}.

\begin{comment}
\section{Uniqueness of solutions}

To fix notation, we define the regularized absolute values as
\begin{align*}
	\regabs{x} \coloneqq (1+\abs{x}^{2})^{1/2}.
\end{align*}
For any given \(p>0\), we can find a suitably small \(c_p \in \R\) such that
\begin{align}
	c_p(1+\abs{x}^p) \leq \regabs{x}^p \leq 1+\abs{x}^p, \quad x \in \R^d
\end{align}
We are not worried about the optimality of this constant, so for our purposes we can pick \(c_p = 1/2c_d\) for every \(p\), where \(c_d\) is some dimensional constant.

Let \(\Phi(x) \coloneqq (\regabs{x}^{p})^{-1}\) and \(X_{p} = X_{\Phi} \coloneqq \{\psi \colon \Z^d \to \C \colon \sup_{x\in \Z^d}\abs{\psi(x)}\Phi(x) <+\infty\}\). We are particularly interested in the case \(p = 1/2\). Each of the aforementioned spaces implies a weighted sup-norm, which we denote by \(\norm{\cdot}_{\infty, p}\). If we restrict our attention to those vectors with \(\norm{\psi}_{\infty, p} \leq M\), we use the notation \(X_{p,M}\). Finally, let
\begin{align*}
	\mathcal{C}_p \coloneqq C([0, +\infty),X_p), \quad \mathcal{C}_{p, T} \coloneqq C([0, T], X_{p}), \quad \mathcal{C}_{p,M, T} \coloneqq C([0,T], X_{p,M}).
\end{align*}
\end{comment}

Let \(\psi_0 \in X_{1/2}\) be an initial value. For each fixed \(T >0 \), we can find a ball of radius \(M < +\infty\) such that \(\norm{\psi_t}_{1/2} \leq M,\) for \(t\in [0,T]\), when \(\psi_t\) is given by any solution to the Cauchy problem \ref{problem:DNLS-cauchy}. This is an a priori boundedness condition that allows us to prove the uniqueness of such solutions.

\begin{lemma}
\label{lemma:uniqueness-of-solutions}
Let \(\psi_0 \in X_{1/2}\) be an initial lattice function \(\psi_0 \colon \Z^d \to \C\), and fix \(T > 0\). If two functions \(\psi^1, \psi^2\) solve the Cauchy problem \eqref{eq:NLS} in the sense of Definition \ref{def:solution-to-DNLS}, then \(\psi^1_t = \psi^2_t\) for all \(t \in [0,T]\) 
\end{lemma}
\begin{proof}
We know that both \(\psi^1\) and \(\psi^2\) satisfy the second order expansion
\begin{align}
\psi_t(x) = \psi_0(x) -\ci t G_x(\psi_0) + \int_{0}^t \rmd s (t-s) P_x(\psi_s), \quad t \in [0, T]	
	\label{eq:integral-formulation}
\end{align}
with the initial data \(\psi^1_{0} = \psi^2_0 = \psi_0\).

If we consider the distance between \(\psi^1_t(x)\) and \(\psi^2_t(x)\), the first two terms always cancel, leaving only the contribution from the integral term:
\begin{align}
\abs{\int_{0}^t \rmd s (t-s)P_x(\psi^1_s) - \int_{0}^t \rmd s (t-s)P_x(\psi^2_s)}.
\end{align}
This term compares \(\psi^1_s\) and \(\psi^2_s\) not only at site \(x\), but also at sites that are inside an \(x\)-centric closed \(\abs{\cdot}_\infty\)-ball with radius \(2\ell\). Now, let's try to find a closed system of estimates for the following quantities that measure how much discrepancy there is between \(\psi_t^1\) and \(\psi_t^2\) inside a centered ball with growing radii: \[\delta_n(t) \coloneqq \sup\{\abs{\psi^1_t(x)-\psi^2_t(x)} \colon \abs{x}_\infty \leq n (2\ell)\}.\]  Indeed, we now see that
	\begin{align}
		\abs{\psi^1_t(x) -\psi^2_t(x)} &\leq \int_{0}^t \rmd s_1 (t-s_1) \abs{\left(P^x(\psi^1_{s_1}) - P^x(\psi^2_{s_1})\right)} \\
		&\leq \int_{0}^t \rmd s_1 (t-s_1)\delta_{n+1}(s_1) c n^2,
	\end{align}
	where \(c\) depends on some combinatorial constants, \(\ell\), \(M\), and the \(X_{1/2}\)-norm of the initial data. The term \(n^2\) comes from the fact that we have a weight \(\regabs{x}^{-1/2}\) in the supremum norm.

	Consequently,
	\begin{align}
		\delta_n(t) \leq \int_{0}^t \rmd s_1 (t-s_1) \delta_{n+1}(s_1)cn^2.
		\label{ineq:iteration-step}
	\end{align}
	We also note that we have the a priori bound \(\delta_n(s) \leq 2MS_{n(2\ell)}\), where \(S_k\) is defined as the supremum of the inverse weight in the centered \(k\)-ball:
	\begin{align} 
	S_{k} = \sup_{\abs{x}_\infty \leq k}\Phi(x)^{-1} = \sup_{\abs{x}_\infty \leq k}\regabs{x}^{1/2} \leq c_d\regabs{k}^{1/2}.
	\end{align} 
	This will be useful after we have iterated the estimate enough times. Indeed, a \(k\)-fold iteration of \eqref{ineq:iteration-step} gives
	\begin{align}
		\delta_n(t) &\leq \int_{0}^t \rmd s_1 (t-s_1) \int_{0}^{s_1} \rmd s_2 (s_1-s_2) \dots \int_{0}^{s_{k-1}} \rmd s_{k}(s_{k-1} - s_k) \delta_{n+k}(s_k)c^k (n\cdots(n+k))^2
		\label{ineq:simplex-integration}
	\end{align}
	After unravelling each \((t-s_i)\) into an integral using \(\int_{0}^t \int_{0}^{s} h(s) \rmd r \rmd s = \int_0^t (t-s)h(s)\rmd s\), we obtain the following bound
	\begin{align}
	\delta_{n}(t) &\leq \frac{c^k t^{2k}}{{(2k)!}} 2MS_{2(n+k)\ell} \left(n(n+1)\cdots(n+k)\right)^2. 
	\end{align}	 
	Now, if we pick
	\begin{align}
	t \leq T_n \coloneqq 1/2c^{-1/2}\left(\limsup_{k\to\infty}\left(2MS_{2(n+k)\ell} \left(n\cdots(n+k)\right)^2(2k)!^{-1}\right)^{1/2k}\right)^{-1},
	\end{align} then the above estimate gives \(\delta_n(t) \leq 1/2\).
Furthermore, if \(t  \leq \overline{T} \coloneqq 1/2 \min(\inf_{n\in \N} T_n, T)\), then \(\delta_n(t)\leq 1/2\) for all \(n\in \N\). After this, the estimate \eqref{ineq:simplex-integration} allows us to bootstrap this, whereby \(\delta_n(s) = 0\) for all \(n\). The key property here is that for each \(n\), and each \(\epsilon > 0\), there exists \(k_0 \geq 1\) such that \(((n+k)!)^2 \leq ((1+\epsilon)k)!^2\) for all \(k\geq k_0\) and showing that for suitably small \(\epsilon\) (here \(\epsilon = 1/10\) is small enough), the right hand side of the following inequality vanishes in the limit \(k \to +\infty\) sufficiently fast:
	\begin{align}
		\frac{(n (n+1)\cdots (n+k))^2}{(2k)!} &\leq \frac{\left(\prod_{\ell=0}^{k+1} (\epsilon k + \ell)\right)^2}{\Gamma(2k+1)}  = \frac{\Gamma((\epsilon +1)k + 1)^2}{\Gamma(\epsilon k)^2 \Gamma(2k+1)}.
	\end{align}
	
	Since \(\delta_n(s) = 0\) for all \(n \in \N\) when \(s\leq \overline{T}\), it follows that \(\psi^1_s = \psi^2_s\) for \(s\leq \overline{T}\).
	
	What is left to show is that \(\overline{T} > 0\). For this to be the case, the expression
	\begin{align}
		\limsup\limits_{k\to\infty} \left(\frac{2MS_{2(n+k)\ell} (n(n+1)\cdots(n+k))^2}{(2k)!}\right)^{1/2k}
		\label{formula:limsup_time_uniform_n}
	\end{align}	 
	has to be uniformly bounded in \(n\). This is the case precisely because we can find a small enough \(\epsilon > 0\) such that the expression \eqref{formula:limsup_time_uniform_n} is bounded from above by 
	\begin{align}
		&\limsup_{k\to\infty}\left( \frac{2MS_{2(n+k)\ell} \Gamma((\epsilon+1)k + 1)^2}{\Gamma(\epsilon k)^2 \Gamma(2k+1)}\right)^{1/2k} \\
		&\leq \limsup_{k\to\infty} \left(2MS_{2(n+k)\ell}\right)^{1/2k} \cdot \limsup_{k\to\infty} \left(\frac{\Gamma((\epsilon+1)k + 1)^2}{\Gamma(\epsilon k)^2 \Gamma(2k+1)}\right)^{1/2k}
	\end{align}	
	The rightmost term is bounded from above, so all we have left to do is to show that there exists \(C <+\infty\) such that
	\[\limsup\limits_{k\to\infty}S_{2(n+k)\ell}^{1/2k} \leq C < +\infty.\] But this is clear, since
	\[S_{2(n+k)\ell}^{1/2k} \leq ((2(n+k)\ell)^{1/2})^{1/2k},\] and for sufficiently large \(k\) we can estimate \(n\leq k\). This gives
\[\limsup_{k\to\infty} S_{2(n+k)\ell}^{1/2k} \leq \limsup_{k\to\infty} ((4k\ell)^{1/2})^{1/2k} \leq C.\]
	We have therefore shown that \(\overline{T} > 0\). Since \(\overline{T}\) only depends on \(M\), the radius of the \(\norm{\cdot}_{1/2}\)-ball where the solution stays for \(t \in [0,T]\), we can use standard pasting arguments to show uniqueness on \([0, T]\).
\end{proof}

\section{Approximation of the limit}
\label{sec:approximation-of-dynamics}

We have shown that the full dynamics \eqref{eq:NLS} exists and is unique for regular enough initial data. Existence was obtained through a compactness argument, and in proving uniqueness we used the second order expansion of the solution to control the difference between two different solutions. What is still left to show is that the full dynamics can in a suitable sense be approximated in terms of the finite periodizations of the dynamics.

Here we have to impose more restrictions on the space of solutions. More precisely, we need to consider a supremum norm with a more restrictive weight. To this end, we use the weight 
\begin{align}
\Phi(x) = \frac{1}{\regabs{x}^{1/2-c}}, \quad x \in \Z^d,
\end{align} where \(c > 0\) can be arbitrarily small. In other words, we work in the space \(X_{1/2-c}\), so that \(\norm{\psi}_{1/2-c} <+\infty\) means that 
\begin{align}
	\abs{\psi(x)}\leq C\regabs{x}^{1/2-c}.
\end{align}
It seems that we really need to go below \(1/2\) in order to get the approximation property to work. In other words, there has to be a slightly tighter regularity restriction on the initial data than for the uniqueness.

As in the proof of uniqueness, given any initial data and a target time \(T > 0\) we may always choose \(M\) so that \(\norm{\psi_t}_{1/2-c} \leq M\) for all \(t\leq T\).

Let \(\psi_0 \in X_{1/2-c}\). We can restrict this initial data to \(\Lambda_L\) and then extend it back periodically to \(\Z^d\) in order to obtain a valid periodic initial data \(\psi^L_0\). For \(x\in \Z^d\) and \(t\in [0, T]\), we would like to define the limiting object of \(\psi^L_t\) as
\begin{align}
	\phi_t(x) \coloneqq \lim_{L\to\infty} \psi_t^L(x), \quad x \in \Z^d.
	\label{eqdef:limit-in-L}
\end{align} We don't yet know if this definition makes sense, since we only have convergence along a subsequence, as shown in the construction of the solution to the infinite system. We have to prove an additional Cauchy condition on the sequence \((\psi_t^L(x))_{L\geq 1}\) to be able to do this. This is the content of Lemma \ref{lemma:uniform-cauchy-condition} below. After this, we are able to prove the following lemma, which is the main result this section.
\begin{lemma}
	\label{lemma:approximation}
	Let \(\psi_0\in X_{1/2-c}\), then the finite problem \ref{problem:DNLS-finite-cauchy} approximates \ref{problem:DNLS-cauchy} well in the sense that for each \(T>0\), \(k \in \N\) and \(\epsilon >0 \), there exists \(L_0 \geq 1\) such that for all \(L \geq L_0\), we have:
	\begin{align}
		\sup_{s\in [0,T]}\sup_{x\in \Lambda_k} \abs{\psi_s(x)-\psi^L_s(x)} < \epsilon.
	\end{align} Here \(\psi_s(x)\) is the limiting field constructed using the subsequence, and \(\psi^L\) are the unique solutions to the finite Cauchy problem \ref{problem:DNLS-finite-cauchy} with initial data \(\psi^L_0\).
\end{lemma}
\begin{remark}
	It follows from this result that we not only have a convergence through a subsequence of the truncated problem, but also that the stronger result 
	\begin{align}
		\psi_t(x) = \lim_{L\to+\infty}\psi^L_t(x), \quad x \in \Z^d, \quad t \in \R_+
	\end{align} holds. Thus, the limiting object in \eqref{eqdef:limit-in-L} exists and is \(\psi_t\).
\end{remark}

Following the strategy of \cite{butta2007propagation} we define, for fixed \(L\in \N\), for \(k \in \N\) with \(k \leq L\), and \(t\in [0, T]\), the following quantities
\[\Delta_{x}^L(t) \coloneqq \sup_{s\in [0, t]}\abs{\psi_s^{L+1}(x) - \psi^L_s(x)}, \quad x \in \Lambda_L\]
\[\overline{\Delta}_{k}^L(t) \coloneqq \sup_{x\in \Lambda_k}\Delta_x^L(t),\]
and
\[d^L(t) \coloneqq \sup_{s\in[0,t]}\sup_{x\in \Lambda_L}\abs{\psi_s^L(x)-\psi_0^L(x)}.\]
The first quantity measures the supremum-in-time distance of the values  of the \(L+1\)-periodized field and the \(L\)-periodized field at site \(x\). Initially, the two periodizations agree for \(\abs{x}_\infty \leq L\), but as soon as the dynamics are turned on, we might start to see growth in \(\delta_x^L(t)\). The second one looks at the supremum of these differences inside some smaller box of side length \(k\leq L\). Finally, the third quantity looks at the how much the values of an \(L\)-periodized solution can deviate from the initial value, when we follow the trajectory up to time \(t \geq 0\). 

\begin{comment}
\begin{lemma}
	\begin{align}
	\abs{\dot{\psi}_t^{L+1}(x) - \dot{\psi}_t^L(x)} &\leq \abs{\psi_0^{L+1}(x)-\psi_0^L(x)} \\
	&+ \int_{0}^t \rmd s \abs{\left(G_x(\psi_s^{L+1}) - G_x(\psi_s^L) \right)} \\
	&\leq \abs{\psi_0^{L+1}(x)-\psi_0^L(x)} \\
	&+ \sup_{x\in \Lambda_L} \abs{\partial_{\psi_t(x)} G_x()}
	\end{align}
\end{lemma}

\begin{lemma}
If \(\psi_0 \in \mathcal{C}_M\), we have \[\abs{\dot{\psi}^{L+1}_t(x)- \dot{\psi}^{L}_t(x)} \leq M\exp(cT/2)\sup_{\abs{y-x}\leq \ell}(1+\abs{y}^{1/2})\delta_y^L(t), \quad \abs{x}_\infty\leq L-\ell\] for all \(t\leq T\).
\end{lemma}

\begin{lemma}
	For \(\psi_0 \in \mathcal{C}_M\), we have
	\[d^L(t) \leq MC_{c,r} t^2 \exp(cT/2)\norm{\psi_0}^4_{U_{r/4}} \leq MC_{c,r} T^2\exp(T)\norm{\psi_0}^{4}_{U_{r/4}}\] for \(t \leq T\).
\end{lemma}

\begin{proof}
	This follows immediately from the equicontinuity bound and \eqref{ineq:approximation-uniform-bound}
\end{proof}
\end{comment}

\begin{lemma}
\label{lemma:uniform-cauchy-condition}
For any initial data \(\psi_0\) and target time \(T>0\), we have that for each \(k\in\Z\), and for each \(\epsilon > 0\), there exists \(L_k \geq 1\) and a constant \(A \geq 2\) such that
\begin{align}
	\overline{\Delta}_{k}^{L}(t) < A^{-L} \quad t\leq T,
\end{align}
for all \(L \geq L_k\)
\end{lemma}
\begin{comment}
\begin{remark}
	In particular, this shows that the limit in \eqref{eqdef:limit-in-L} exists for each \(x \in \Z^d\), and this limit is uniformly approximated by \(\psi^L\) inside any finite box. In other words, for a fixed time \(T\), an error of size \(\epsilon > 0\), and \(k\in \N\), we can find a side length \(L_k\) such that the full dynamics are within an error of \(\epsilon\) from the periodized dynamics uniformly inside a finite box of size length \(k\) up to times \(T\). Essentially, this means that there is a finite speed of sound in the system, since for a fixed time and a fixed error bound, large differences far away from the origin cannot produce large differences near the origin
\end{remark}
\end{comment}
\begin{proof}
	For the sake of simplicity, we assume that the range of the hopping potential is \(\ell = 1\). 
	Assume that \(\abs{x} \leq k\). The first Picard iteration gives 
	\[\psi_t^L(x) = \psi_0^L(x) - \ci \int_{0}^t \rmd s G^x(\psi_s^L),\]
	where \(G^x(\psi_s^L)\) depends only on \(\psi^L_s(y)\) with \(\abs{x-y} \leq 1\), and has a gradient bound controlled by \((1+\abs{x})^{1-2c} \leq (k+1)^{1-2c}\). It follows that
	\begin{align}
		\Delta_x^L(t) \leq C \int_{0}^t (k+1)^{1-2c}
		\overline{\Delta}_{k+1}^{L}(s)\rmd s,
	\end{align}
	whereby
	\begin{align}
	\overline{\Delta}_k^L(t) \leq C \int_{0}^t(1+k)^{1-2c}\overline{\Delta}_{k+1}^L(s) \rmd s
	\end{align}
	We can iterate this \(j\) times as long as \(k + j \leq L\), so at most \(L-k\) times, which yields
	\begin{align}
		\overline{\Delta}_k^L(t) \leq C^{L-k}(1+k)^{1-2c}\dots (1+L)^{1-2c} \int_{0}^t \rmd s_1 \dots \int_{0}^{s_{L-k-1}} \overline{\Delta}_L^L(s_{L-k}) \rmd s_{L-k}. 
	\end{align}
	
	At the final step, we use the fact that \begin{align}
	\sup_{x\in \Lambda_{L}}\left(\abs{\psi_s^{L+1}(x)} + \abs{\psi_s^L(x)}\right) \leq 2C_T\norm{\psi_0}_{1/2-c}\regabs{L}^{1/2-c}.
	\end{align}
	
	This shows that
	\begin{align}
	\overline{\Delta}_k^L(t) \leq C_{T, \norm{\psi_0}_{1/2-c}}^{L-k}\frac{t^{L-k}}{(L-k)!}\left(\prod_{i=1}^{L-k}(k + i)\right)^{1-2c}L^{1/2-c}.
	\end{align}
	
	We can therefore pick a cutoff point \(L_k\) depending on \(\epsilon, k, \norm{\psi_0}_{1/4}\), and \(T\), such that
	\[\overline{\Delta}_k^L(T) < A^{-L}\] for all \(L \geq L(k, \norm{\psi_0}, T)\), where \(A \geq 2\).
	\end{proof}
	We are now ready to prove the approximation result.
	\begin{proofof}{Lemma \ref{lemma:approximation}}
		Fix \(T > 0\), \(k \in \N\) and \(\epsilon > 0\). We can find a cutoff \(\tilde{L}\) such that
		\begin{align}
			\sup_{s\in [0, T]}\sup_{x\in \Lambda_k} \abs{\psi_s^{L+1}(x)-\psi_s^{L}(x)} < A^{-L}, \quad L \geq \tilde{L},
		\end{align}
		and the tail \(\sum_{k=1}^{\infty} A^{-(L+k)}\) sums to \(<\epsilon/2\) (being a tail of a geometric series).
		
		On the other hand, we can find \(L_0\) from the subsequence along which we converge to \(\psi\) uniformly on \(\Lambda_k\), \(s\in [0, T]\), such that
		\begin{align}
			\sup_{s\in [0, T]}\sup_{x\in\Lambda_k}\abs{\psi_s(x)- \psi_s^{L_{\ell}}(x)}< \epsilon/2, \quad L_{\ell}>L_0. 
		\end{align}
		Thus,
		\begin{align}
			 \sup_{s\in [0, T]}\sup_{x\in \Lambda_k} \abs{\psi_s(x) -\psi_s^L(x)} < \epsilon/2 + \epsilon/2 = \epsilon
		\end{align}
		for all \(L \geq \tilde{L}\).
	\end{proofof}

\section{Random initial data}
\label{sec:random-initial-data}

Theorem \ref{thm:existence-uniqueness-approximation} shows that for initial data that is in \(X_{1/2-c}\), the \eqref{problem:DNLS-cauchy} has a unique solution that can be approximated in terms of solutions to \eqref{problem:DNLS-finite-cauchy}. In this Section, we investigate whether the certain states -- probability distributions on the configuration space -- assign full probability to \(X_{1/2-c}\).

We will only consider translation invariant probability measures on \(\C^{\Z^d}\), since all equilibrium states satisfy this property and since the Gaussians used in studying homogeneous kinetic theory of waves are translation invariant. Translation invariance makes the use of the following lemma from \cite{lanford1977time} to use the uniform bounds for certain moments to establish almost sure growth conditions for the samples.

\begin{proposition}
	\label{proposition:lanford-et-al}
	Let \(\mu\) be a probability measure on \(\C^{\Z^d}\). For a fixed \(a>0\), if there exists \(\xi > a\cdot d\) satisfying
	\begin{align}
		\sup_{x\in \Z^d}\int_{\C^{\Z^d}} \abs{\psi(x)}^{\xi}\mu(\rmd \psi) <+\infty,
		\label{ineq:uniform-moment-bound}
	\end{align} then 
	\begin{align}
		\mathbb{P}(\limsup_{n\to\infty}E_n) = 0,
	\end{align}
	where
	\begin{align}
		E_{n} \coloneqq \{\psi \colon \Z^d \to \C \colon \abs{\psi(x_n)}\regabs{x_n}^{-1/a} > 1\}.
	\end{align}
	Here \((x_n)_{n=1}^\infty\) is an enumeration of all \(x \in \Z^d\).
\end{proposition}
\begin{remark}
	This is stated and proved in \cite[pp. 460--461]{lanford1977time}.
\end{remark}

In order to apply this proposition to our setting and get the full existence, uniqueness and approximability, we can take any \(a > 2\). Thus, we need need to show that with \(\xi > 2d\), something like \eqref{ineq:uniform-moment-bound} holds. Since \(d \geq 1\), we will always need to go beyond \(\xi =2\). On the other hand, if we want to obtain only existence, then any \(a > 0\) will do, and it suffices to prove that \eqref{ineq:uniform-moment-bound} holds for, say \(\xi = 2\), which is typically easy.

\begin{theorem}
	\label{thm:almost-sure-initial-data}
	If \(\mu\) is a measure on \(\C^{\Z^d}\) from one of the following classes, then with probability \(1\) with respect to \(\mu\) a sequence is in \(X_{1/2-c}\).
	\begin{itemize}
		\item[(1)] The measure \(\mu\) is a centered, translation invariant Gaussian process on \(\Z^d\).
		\item[(2)] \(d = 1\) and the measure \(\mu\) is a microcanonical measure for the Hamiltonian system \eqref{eq:NLS} with particle density \(\rho_N\) and energy density \(\rho_E\).
		\item[(3)] \(d=1\) and the measure \(\mu\) is a grand canonical Gibbs measure for the Hamiltonian system \eqref{eq:NLS}.
	\end{itemize}
\end{theorem}
Before proving the theorem, let us briefly mention what the microcanonical and Gibbs measures of the system would look like. These are both defined as limits of certain finite dimensional probability measures. In what follows, we use the notation \(V_d = (2L+1)^d\) to stand for the total volume of a \(d\)-dimensional box of side-length \(2L+1\). The notation \(H_L(\psi)\) will stand for the Hamiltonian \eqref{eq:finite-hamiltonian} and \(N_L(\psi)\) will stand for the particle number, i.e. the square of \(\norm{\cdot}_{\ell^2(\Lambda_L)}\).

\begin{definition}[Microcanonical measure]
	Let \(\rho_E\) and \(\rho_N\) be the energy and particle density, respectively. Then, provided that the Diracs make sense, the measures \(\mu_{L,\rho_E,\rho_N}\) on \(\C^{\Lambda_L}\), are defined as
	\begin{align}
		\mu_{L,\rho_E,\rho_N}(\rmd \psi) \coloneqq Z_{L,\rho_E,\rho_N}^{-1}\delta(H_L(\psi)-\rho_EV_d)\delta(N_L(\psi)-\rho_N V_d)\prod_{x\in \Lambda_L}\rmd(\Re \psi(x))\rmd(\Im \psi(x))
		\label{eq:def-microcanonical-finite}
	\end{align}
	are the microcanonical measures of the finite system. Here \(Z_{L,\rho_E,\rho_N}\) is a normalization constant that makes \eqref{eq:def-microcanonical-finite} a probability measure.
	
	Provided that the limit \(\lim_{L\to\infty}\mu_{L,\rho_E,\rho_N}\) exists in some sense, we call this a microcanonical measure of the infinite system \eqref{eq:NLS}.
\end{definition}

\begin{definition}[Grand Canonical Gibbs measure]
	Let \(\beta > 0\) be the inverse temperature and \(\nu \in \R\) the chemical potential. For each \(L\), we define the Gibbs measures of the finite DNLS as
	\begin{align}
		\nu_{L,\beta,\mu}(\rmd \psi) = Z^{-1}_{L,\beta, \mu} \rme^{-\beta(H_L(\psi)-\mu N_L(\psi))} \prod_{x\in \Lambda_L} \rmd (\Re \psi(x))\rmd (\Im \psi(x)).
		\label{eq:def-gibbs-measure}	
	\end{align} 
	Here \(Z_{L,\beta,\mu}\) is a normalization constant that makes \eqref{eq:def-gibbs-measure} a probability measure.
	
	Provided that the limit \(\lim_{L\to\infty}\nu_{L,\beta, \mu}\) exists in some sense, we call this a Grand canonical Gibbs measure of the infinite system \eqref{eq:NLS}.
\end{definition}
The question of the existence or nonexistence of the Gibbs measures for infinite systems is a delicate matter and beyond the scope of this article. This should depend on the dimension of the problem, as well as the size of the nonlinearity in \(H_L\) and the negativity of the chemical potential \(\mu\). For this reason, we will only state conditional results that any such grand canonical Gibbs measure must satisfy. This can be achieved by getting uniform estimates for the finite versions of the measures.

We are now ready to prove Theorem \ref{thm:almost-sure-initial-data}.
\begin{proof}
	To prove the Gaussian case, we note that if we have a centered, translation invariant Gaussian process on \(\Z^d\), then for any \(x \in \Z^d\) and any \(\xi > 2d\), the moment
	\begin{align}
		\int_{\C^{\Z^d}} \abs{\psi(x)}^{\xi}\mu(\rmd \psi) \leq C
		\label{bound:uniform-moment-ell}
	\end{align}
	is bounded. By translation invariance, it follows that
	\begin{align}
		\sup_{x\in \Z^d}\int_{\C^{\Z^d}} \abs{\psi(x)}^{\xi} \mu(\rmd \psi) \leq C.
	\end{align}
	Thus, by the result of Lanford et al., it follows that \(X_{1/2-c}\) has full probability.
	
	For item (2), we note that \(N_{L}(\psi)\) and \(H_L(\psi)\) are translation invariant for each \(L\), we have
	\begin{align}
		\int_{\C^{\Lambda_L}} \abs{\psi(x)}^2 \mu_{L, \rho_E,\rho_N}(\rmd \psi) = \frac{1}{V_d} \rho_N V_d = \rho_N.
	\end{align}
	Similarly,
	\begin{align}
		\int_{\C^{\Lambda_L}} H_L(\psi)\mu(\rmd \psi) &=
		\sum_{x \in \Lambda_L} \int_{\C^{\Lambda_L}} \sum_{y \in \Lambda_L} \alpha(x-y)\psi(x)\psi(y)^* \mu_{L,\rho_E,\rho_N}(\rmd \psi) \\
		&+ \sum_{x \in \Lambda_L} \int_{\C^{\Lambda_L}} \abs{\psi(x)}^4 \mu_{L,\rho_E,\rho_N}(\rmd \psi) \\
		&= \rho_EV_d
	\end{align}
	On the other hand,
	\begin{align}
		\sum_{x,y \in \Lambda_L^2} \alpha(x-y) \psi(x)\psi(y)^*  = \sum_{n = 1}^N \inner{\psi}{S_n \psi}.
	\end{align}
	Here each \(S_n\) is a self-adjoint operator on \(\ell^2(\Lambda_L)\) that has uniform operator norm with respect to \(L\). It follows that
	\begin{align}
		\sum_{x\in \Lambda_L} \int_{\C^{\Lambda_L}} \abs{\psi(x)}^4 \mu_{L,\rho_E,\rho_N}(\rmd \psi) &\leq \rho_EV_d - C\int_{\C^{\Lambda_L}}\norm{\psi}_{\ell^2}^2 \mu_{L,\rho_E,\rho_N}(\rmd \psi) \\
		&= \rho_E V_d - C \rho_NV_d.
	\end{align}
	By translation invariance of \(\mu_{L,\rho_E,\rho_N}\), we have.
	\begin{align}
		\int_{\C^{\Lambda_L}} \abs{\psi(x)}^4 \mu_{L,\rho_E,\rho_N}(\rmd \psi) \leq C',
	\end{align}
	where \(C'\) is uniform in \(L\). Thus, if  \(\mu_{\rho_E,\rho_N}\) exists, any moment \(\xi < 4\) must be bounded, which through the Proposition \ref{proposition:lanford-et-al} of Lanford et al. implies that \(X_{1/2-c}\) has full measure when \(d=1\).  
	
	For item (3), note (see \cite{dodson_nonlinear_2020}) that any accumulation point of the periodic Gibbs measures must satisfy the boundedness of \eqref{ineq:uniform-moment-bound} for every \(\xi < 4\). Thus, if \(d=1\), the result is strong enough to conclude by the above argument that \(X_{1/2-c}\) has full probability with respect to such Gibbs measure.
\end{proof}

\bibliography{references}
\end{document}